\documentclass[sigconf]{acmart}

\usepackage{bbm}
\usepackage{graphicx}
\usepackage{subfigure}
\usepackage{float}
\usepackage{amsfonts}
\usepackage{enumerate}
\usepackage{algorithm}
\usepackage{algpseudocode}
\usepackage{amsmath}
\usepackage{enumitem}
\usepackage{natbib} 
\usepackage{caption}
\usepackage{hyperref}
\usepackage[marginal]{footmisc}
\usepackage{appendix}
\usepackage{amsmath}

\setcopyright{rightsretained}
\theoremstyle{plain}
\newtheorem{theorem}{Theorem}
\newtheorem{lemma}{Lemma}

\newtheorem{assum}{Assumption}
\newtheorem{assumption}{Assumption}
\newtheorem{definition}{Definition}

\newtheorem{rmk}{Remark}
\renewcommand{\vec}[1]{{\boldsymbol{{#1}}}}
\usepackage{tikz}
\usetikzlibrary{backgrounds}
\usetikzlibrary{fit,positioning}
\tikzstyle{vertex}=[circle, draw, inner sep=1pt, minimum size=15pt]

\begin{document}
\title{Model Free Reinforcement Learning Algorithm for Stationary Mean field Equilibrium for Multiple Types of Agents}
\author{Arnob Ghosh}
\email{arnob.ghosh@imperial.ac.uk}
\affiliation{%
  \institution{Imperial College of London}
  \city{London}
 \country{UK}
}
\author{Vaneet Aggarwal}
\email{vaneet@purdue.edu}
\affiliation{%
  \institution{Purdue University}
  \city{West Lafayette}
  \state{IN}
  \country{USA}
}

	\begin{abstract} 
	We consider a multi-agent Markov strategic interaction over an infinite horizon where agents can be of multiple types.  We model the strategic interaction as a mean-field game in the asymptotic limit when the number of agents of each type becomes infinite. Each agent has a private state; the state  evolves depending on the distribution of the state of the  agents of different types and the action of the agent. Each agent wants to maximize the discounted sum of rewards over the infinite horizon which depends on the state of the agent and the distribution of the state of the leaders and followers. We seek to characterize and compute a stationary  multi-type Mean field equilibrium (MMFE) in the above game. We characterize the conditions under which a stationary MMFE exists.  Finally, we propose Reinforcement learning (RL) based algorithm using policy gradient approach to find the stationary MMFE when the agents are unaware of the dynamics. We, numerically, evaluate how such kind of interaction can model the cyber attacks among defenders and adversaries, and show how RL based algorithm can converge to an equilibrium.
\end{abstract}
	\maketitle
	

	\section{Introduction}
	In many real world applications, such as in cyber-physical system, intelligent transportation, cyber-security, smart grid, and Internet of Things (IoT), strategic agents (e.g., smart devices, robots, rational humans) interact with each other. In most of these applications, the agents interact repeatedly and the `reward' or the utility of an agent not only depends on her own action but also on the actions of other agents. 
	
	Finding the optimal strategy of the agents becomes more challenging when the agents are unaware of the exact environment and reward.  Multi-agent Reinforcement learning (MARL) has been developed to model the uncertainty \cite{hernandez2019survey}. In  MARL, similar to a Markov stochastic game, the environment of an agent is embedded as a state. The state and reward obtained by an agent depends on the current state or action of each individual agent.   However, due to the curse of dimensionality, analyzing MARL with a large number of players is challenging when the agents interact for a longer time. 
	
	Mean-field game (MFG) has recently been popular to study interaction among a large number of {\em statistically similar} agents. MFG studies the game in the asymptotic limit where the number of agents is considered to be infinite. When the number of agents becomes large, only the distribution of the states (or, actions) of all the agents rather than the individual state (or, action) along with her own state and action only impacts the state and reward of an individual agent in a MFG. Thus, in a MFG, each agent only needs to consider the distribution of the agents' states (or, actions) rather than individual state or action of the other agents. MFG along with the Reinforcement learning setting has been of recent interest to analyze strategic interactions among a large number of agents in an unknown environment.
	
	In the traditional MFG, all the players are statistically similar. However, in a lot of scenario players can be of different types. In cyber-physical system, agents are often of different types.  Further, leader-follower type of interaction where  a set of agents (leaders) first take their actions and the other set of agents (followers) follow in a MFG setting  is ubiquitous.   For example, in cyber-security, defenders protect the system against the adversaries. The defenders try to protect the system by installing firewalls, the adversaries then try to breach those systems which are not well protected. Further, 
	
	Analyzing those games in a MFG setting is challenging.  The agents are {\em not statistically} similar any more since there are now two different classes of agents: leaders and followers. In a MFG, mean-field equilibrium (MFE) policy is defined as an equilibrium policy where no agent has any incentive to deviate from its own strategy. However, when agents are of different types the equilibrium needs to be optimal In a leader-follower game with large number of agents, we need to compute a MFE for followers for any policy of a leader, then we need to determine the MFE policy for the leaders which must also consider the MFE policy for the followers. The analysis becomes more challenging when the agents are unaware of the reward functions and the transition probability kernels. 
	
	\subsection{ Contribution}
	We, first, develop theoretical tools for modeling and analyzing a simultaneous Markov MFG (Section~\ref{sec:mfg}) where agents of different types interact simultaneously over a long period of time.  We propose a Multi-type Mean-field-equilibrium (MMFE) concept in the MFG setting (Section~\ref{sec:mfe}). In a MMFE, there are two policies--one for each type of agents. In a MMFE, the policy at a given state of an agent must be optimal given the actions of the agents of the same type as well as the actions of the agents of different types. Further, the evolution of the distribution of the states of both the types of agents must be consistent with the policies. 
	

	We consider a stationary MMFE where the population distribution of the agents, and the policies of the agents do not change over time (Section~\ref{sec:stationary}). Thus, in the stationary MMFE, the population distribution at the next time must remain the same even when the agents of each type take their optimal actions based on individual state.    We propose an algorithm which converges to the stationary MMFE when the leaders and followers know the dynamics of the Markov Decision Process (Section~\ref{sec:compute}). In the algorithm, first for a given population distribution of both the types, one finds the optimal policies for agents of each type. The population distribution is updated based on the optimal policies. The steps are continued till the population distributions become the same. We characterize the condition under which a fixed point exists.  
	
	We, subsequently, provide an adaptation of the algorithm for the case where  the reward and transition probability kernels are unknown to the agents using the RL approach (Section~\ref{sec:rl}). Specifically, a policy gradient based algorithm is used to compute an optimal policy for the agents for a given population distribution. We characterize the conditions on the policy parameters such that the algorithm converges. Numerically, we model the interaction between defenders and attackers in a cyber security space as a Multi-agent Markov game and shows that our RL based algorithm converges to the MMFE (Section~\ref{sec:numerical}). 
	
	
	\subsection{Related Literature}
	Multi-agent interaction using a Markov game has been studied \cite{hu,littman,Bua,tan}. Recently, multi-agent reinforcement learning using mean field game has also been studied  \cite{yang,guo}. In \cite{guo1}, the authors characterized the conditions under which a stationary MFE exists and proposed algorithm to obtain that.  In \cite{mahajan}, authors proposed a policy gradient mechanism to compute a local stationary MFE. In \cite{carmona2019modelfree}, authors proposed a model free reinforcement learning algorithm and showed that the algorithm converges to a MFE under some regularity conditions. In \cite{carmona2019linearquadratic}, the authors showed that a policy gradient mechanism for computing MFE in a linear quadratic controller. 
	
	In \cite{fu2019actor}, an actor-critic method has been proposed to determine MFE.  \cite{mguni2018decentralised} shows that MFG is a potential game for a strictly monotone game. \cite{perrin2020fictitious} shows that under certain conditions, fictitious play converges to a MFE. \cite{yang2017learning} proposed a deep mean field game to model strategic interactions among the players. \cite{mishra2020modelfree} proposed a model-free non stationary algorithm to compute MFE. Compared to all the above literature,we focus on leader-follower kind of interaction in a MFE where both the leaders and followers are  different. Further, we proposed a policy-gradient type algorithm which is faster in convergence compared to $Q$-learning method adopted in \cite{guo,guo1,carmona2019modelfree}. Compared to \cite{mahajan}, we characterize the conditions under which as stationary MMFE exists which are easy to verify. Further, we propose a novel Random-Horizon policy gradient mechanism for obtaining MMFE. 
	
	Leader-follower type game in  a game theoretic setting is extensively studied \cite{moore1988subgame,harris1995existence}. Hierarchical Markov game has also been studied. \cite{10.5555/2891460.2891644,4220822} considered a tractable leader-follower MDP game. \cite{sengupta2020multiagent} considered a Stackelberg Markov game and proposed a Stackelberg $Q$-learning algorithm to compute  $Q$-values for leaders and followers. However, all the above papers considered only one leader. Compared to these existing literature, we consider a MFG where agents can be of multiple types. We consider MMFE as an equilibrium concept and propose a Reinforcement learning algorithm to compute the equilibrium where the agents are unaware of the dynamics. 

\section{Simultaneous Multi-agent Game}
We, first, define a Markov game where agents of different types. Subsequently, we define the mean field game equilibrium for  game when the agents of each type becomes infinite. 
\subsection{Background: Multi-types Multi-agent Markov Game}
We consider a classical Markov game with $N_j$ agents of type $j$.  In Section~\ref{sec:discussion}, we discuss how to extend the model when there are more than two types. At each time $t$, the state of an agent $i$ of type $j$ is $x_{j,t}^i\in \mathcal{X}_j$ and she takes an action $a_{j,t}^i\in \mathcal{A}_j$ for $i=1,\ldots,N_j$. 

At time step $t$, the agent $i$ of type $j$ gets a reward based on her own state and action, the other agents' states of type $j$, and the states of agents of type $k\neq j$. We denote the reward of any agent $i$ of type $j$ as $r^j(\mathbf{x}_{1,t},\mathbf{x}_{2,t},a_{j,t}^i)$. Here, $\mathbf{x}_{j,t}$ is the state vector of all the agents of type $j$. All the agents take actions simultaneously.\footnote{It is equivalent to the setting where the agents do not observe the actions of the other agents.}    



The state of agent $i$ of type $j$ at time $t+1$ evolves depending on the transition probability kernel $\tau^j(\cdot|\mathbf{x}_{1,t},\mathbf{x}_{2,t},a_{j,t}^i)$. The transition probability of agent $i$ of the type $j$ depends on the state vector of all the agents of both the types and action of the agent $i$, $a_{j,t}^i$.  Note that inherently, the transition probability depends on the state of the agent $i$ which is included in the joint state space of the agents.  
\begin{rmk}\label{rmk:state}
The reward and the transition probabilities may also depend on the joint actions of agents of different types.  We can extend our analysis to the above scenario  and we will discuss it in Section~\ref{sec:discussion}. 
\end{rmk}

A Markovian game restricts the admissible policies for a player to be of markovian in nature. The policy for an agent $i$ of type $j$ is given by $\pi^{i,j}_t:\mathcal{X}_1^{N_1}\times \mathcal{X}_2^{N_2}\rightarrow \mathcal{P}(\mathcal{A}_j)$, where $\mathcal{P}(\mathcal{A}_j)$ is the probability space of the action space of agents of type $j$. Thus, the policy of an agent depicts the action which the agent takes given the state vector of all the agents. Note that the policy can be randomized where an agent can randomize over multiple actions. A policy can also be deterministic where the agent takes a certain action for a specific joint state distribution of leaders and followers. 


The accumulated reward for an agent $i$ (a.k.a. value function) of type $j$, starting from an initial state vector $\vec{x}_{0,1}, \vec{x}_{0,2}$  and the policy sequence $(\vec{\pi}^{1},\vec{\pi}^2)=\{\vec{\pi}^1_t, \vec{\pi}^2_t\}^{\infty}_{t=0}$ with $\vec{\pi}^1_t=\{\pi_t^{1,1},\ldots,\pi_t^{N_1,1}\}$ and $\vec{\pi}^2_t=\{\pi_t^{1,2},\ldots,\pi_t^{N_2,2}\}$ is given by
\begin{equation}
    V^j_i(\vec{x}_1,\vec{x}_2,\vec{\pi}^1,\vec{\pi}^2)=\mathbb{E}[\sum_{t=0}^{\infty} \gamma^tr^j(\vec{x}_{1,t},\vec{x}_{2,t},a_{j,t}^i)|\vec{x}_{0,1}=\vec{x},\vec{y}_{0,1}=\vec{y}]\nonumber
\end{equation}
where $0<\gamma<1$ is a discount factor, $a_{j,t}^i\sim\pi^{i,1}_t(\vec{x}_{t,1},\vec{x}_{t,2})$ and $x^i_{j,t+1}\sim \tau^j(\vec{x}_{t,1},\vec{x}_{t,2},a_t^i), \vec{x}_j=\{\vec{x}_{j,t}\}^{\infty}_{t=0}$. Each agent $i$ of type $j$ would want to maximize her own value function over the policies $\pi^{i,j}$.  The expectation here is both on the action profile of all the agents and the transition probabilities based on the action profiles of agents.

\subsection{Mean Field Game for different types of agents}\label{sec:mfg}
In general, Markovian game where the agents are of single type is difficult to analyze for finite number of players. Mean-field-game (MFG), pioneered by \cite{huang2006} and \cite{lasry2007mean} in the continuous setting and letter developed for the discrete setting \cite{BENAIM2008823,lapez, basar}, provides a tractable approach for analyzing the finite player Markovian game. The idea of MFG is simple: Assume that all the players are identical, interchangeable, and indistinguishable from one other, when the number of player $N\rightarrow \infty$, we can view the limit of player states as a population state distribution. 

In the  MFG of multiple types of agents, similar to the MFG where agents are of single type, we consider that both $N_1\rightarrow \infty$ and $N_2\rightarrow \infty$. The agents of type $j$ are identical and indistinguishable from each other. {\em However, the agents of types $1$ and $2$ are different.} For each agent's  perspective, the distribution of the states of all the other agents of its own type, states of the agents of other type, and her own state and action would impact the states and rewards.

We define the limit of the state distribution of the agents of type $j$ as $z_{j,t}$
\begin{equation}
z_{j,t}(s)=\lim_{N_j\rightarrow \infty}\dfrac{\sum_{i=1}^{N_j}\mathbbm{1}_{x_{j,t}^{i}=s}}{N}
\end{equation}
where $\mathbbm{1}_{x_{j,t}^{i}=s}$ indicates that the state of agents $i$ of type $j$ at time $t$ is $s$ (the value is $1$ if the state of agent $i$ of type $j$ at time $t$ is $s$ and is zero otherwise). This is also known as {\em type $j$'s population distribution.}


The transition probability kernel of an agent $i$ of type $j$ when she is in state $x_{t,j}$ and takes an action $a_{t,j}$ is now represented as $\tau^j(\cdot|x_{j,t},z_{1,t},z_{2,t},a_{j,t})$.    The reward function of agent $i$ of type $j$ when she is in state $x_{j,t}$ and takes action $a_{j,t}$, is represented by $r^j(x_{j,t},a_{j,t},z_{1,t},z_{2,t})$. Hence, the transition probability and the reward both depend on the current state, action and the population distribution of the agents of both the types.  Note that the reward and the transition probabilities are identical for each agent of type $j$. However, the reward and the transition probabilities are different for agents of different types. Note that the agents of similar types may have greater influence on the agent compared to the agents of different type. Hence, the transition probability and rewards of agents of type $j$ may be influenced by $z_{j,t}$ compared to $z_{k,t}$ for $k\neq j$. 

The policy of an agent of type $j$ is now a function from state $x_{j,t}$, the population distribution of type $j$'s agents  $z_{j,t}\in \mathcal{Z}_j$, the population distribution of the type $k$s agents $z_{k,t}\in \mathcal{Z}_k$ $k\neq j$ to the probability space over the actions. It is a mapping from $\pi^j_t:\mathcal{Z}_1\times\mathcal{Z}_2\times \mathcal{X}_j\rightarrow P(\mathcal{A}_j)$

\begin{rmk}
We can represent $z_t=\{z_{1,t},z_{2,t}\}$ as the population state distribution of both types of agents.
\end{rmk}


Due to the homogeneity among the agents of particular type, we consider a representative agent who wants to maximize the value function starting from an initial state $x$:
\begin{align}
    V^j(x,\vec{\pi}^1,\vec{\pi}^2,\vec{z}_{1},\vec{z}_{2})=\mathbb{E}[\sum_{t=0}^{\infty} \gamma^tr^j(x_{j,t},z_{1,t},z_{2,t},a_{j,t})|x_{j,0}=x],
\end{align}
where $\vec{z}_j$ is the population distribution sequences of type $j$  i.e., $\vec{z}_j=\{z_{j,t}\}^{\infty}_{t=0}$. Here the expectation is over both the transition kernel, and the policies $\vec{\pi}^1$ and $\vec{\pi}^2$. 

   
\subsection{Multi-agent  Mean Field Equilibrium}\label{sec:mfe}
Now,  we define the multi-agent mean field equilibrium (MMFE).
\begin{definition}
A player-population profile $(\vec{\pi}^1,\vec{\pi}^2,\vec{z}_{1},\vec{z}_2)$ is called a MMFE if the following hold:
\begin{itemize}
    \item MFE for each agent of type $j$, $J=1,2$: The policy must be MFE for the representative agent of type $j$, i.e.,  it must be optimal for a given reaction policy of agents of other type $k$, $K\neq j$, $\vec{\pi}^k$, i.e., for all state $x$ and $\forall \vec{\pi}^{j,\prime}$
    \begin{align}
        V^j(x,\vec{\pi}^j,\vec{\pi}^k,\vec{z}_{1},\vec{z}_2)\geq V^j(x,\vec{\pi}^{j,\prime},\vec{\pi}^k,\vec{z}_{1},\vec{z}_2)
    \end{align}
    \item Population Consistency: We must have for all $t\geq 0$ and for $j=1,2$
    \begin{eqnarray}
        z_{j,t+1}(\tilde{x})=\sum_{x_j,a_j}\tau^j(\tilde{x}|x_j,a_j,z_{1,t},z_{2,t})
    \pi^j_t(a_j|x_j,z_{1,t},z_{2,t})z_{j,t}(x)\label{eq:pop_consist}
    \end{eqnarray}
\end{itemize}
where $z_{1,0}$ and $z_{2,0}$are initial population distributions of type $1$ and $2$ respectively.
\end{definition}
The first part of the definition specifies that $\pi^j_t$ must be optimal for agents of type $j$ for any given state given the policy of other agents of other type $k\neq j$ for each type $j$.  The population consistency guarantees that the state distribution of the type $j$ $j=1,2$ is consistent with the state distribution evolution of the agents of  type $j$ given the state distributions and policies of the agents for other type $k\neq j$. This part is unique for a MFG setting. 

We represent (\ref{eq:pop_consist})  as
\begin{align}\label{eq:mean}
z_{j,t+1}=\phi^j(\pi^j_t,z_{1,t},z_{2,t})
\end{align}
where $\phi^j$ depends on the transition probability $\tau^j$, the policy $\pi^j$ and $z_{j,t}$ as evident from (\ref{eq:pop_consist}).


\begin{rmk}\label{rm:act-mfe}
If we consider that the reward and transition probability also depend on the action distributions of the agents, the action distribution also needs to be consistent with the MFE policy. Hence, we would need  additional conditions, given as follows for $j=1,2$:
\begin{align}{eq:action}
    \nu_{j,t+1}(a_j)=\sum_{x_j}z_{j,t+1}(x_j)\pi^j_{t+1}(a_j|x_j,z_{1,t+1},z_{2,t+1})
\end{align}
Note that here the action distribution implicitly depends on the action of the agents of type $k$ $k\neq j$, since the transition probability and reward both becomes functions of the action distributions of the agents of type $k$ as well. 
\end{rmk}
\subsection{Stationary MMFE}\label{sec:stationary}
In a stationary MMFE, $z_{j,t}=z_{j}$,  and $\pi^j_t=\pi^j$ $\forall t$, for $j=1,2$. Thus,  for a stationary MMFE, we have,
\begin{definition}
A player-population profile $(\pi^1,\pi^2,z_{1},z_{2})$ is called a  stationary MMFE if the following hold:
\begin{itemize}
    \item MFE for agents of type $j$: The stationary policy must be MFE for the representative agent of type $j$,  i.e., for all state $x$ and $\forall \pi^{j,\prime}$, for $j=1,2$, and $k\in \{1,2\}, k \neq j$
    \begin{equation}
        V^j(x_j,\pi^j,\pi^k,z_{1},z_2)\geq V^j(x_j,\pi^{j,\prime},\pi^k,z_{1},z_2)
    \end{equation}
    \item Population Consistency: We must have for all $ t$ and $j=1,2$
    \begin{equation}
        z_{j}=\phi^j(\pi^j, z_{1}, z_{2}) 
        \label{eq:state_consist}
       \end{equation}
\end{itemize}
\end{definition}
In general, computing a non-stationary MMFE is challenging since one needs to find a fixed point at every time step by considering the sequences of population distributions. Thus, in non-stationary MMFE, a finite horizon is generally used. Stationary MMFE is computationally simple. Thus, in the following we focus on the conditions under which a MMFE exists and how to compute a stationary MMFE.

\subsection{Bellman Equation for a stationary MMFE}
In this section, we represent the Bellman equation corresponding to the value function which we use in our proposed Algorithm and the corresponding analysis in Section~\ref{sec:rl}. 

%

Note that in the stationary MMFE, since the population distribution is constant, thus, the action policy of the agents of the other type do not impact the value function in the equilibrium. Thus, we omit $\pi^k$ from the value function for type $j$.  Also even though at the equilibrium, the value function of an agent of type $j$ is independent of $\pi^k$, the equilibrium policy of agent $j$ will depend on the equilibrium policy for agents of type $k$ since the equilibrium policies induce the population distributions $z_k$. The Bellman equation for an agent of type $j$, $j=1,2$ at the equilibrium can be computed as
\begin{eqnarray}
     V^{j}(x_j,\pi^j,z_{1},z_{2})
     =&\sum_{a_j}[\pi^j(a_j|x_j,z_{1},z_{2})(r^j(x_j,z_{1},z_{2},a_j)\nonumber\\
     &+\gamma(\sum_{x^{\prime}}\tau^j(x^{\prime}|x_j,a_j,z_{1},z_{2})\nonumber\\
     &\mathbb{E}[\sum_{k=1}^{\infty} \gamma^{k-1}r^j(x_{j,k},z_{1},z_{2},a_{j,k})|x_{j,1}=x^{\prime}])])\nonumber\\
     =&\sum_{a_j}\left[\pi^j(a_j|x_j,z_{l},z_{f})(r^l(x,z_{l},z_{f},a_j)\right.\nonumber\\
     &+\left.\gamma\mathbb{E}\left[ V^{l}(x^{\prime},\pi^j,z_{1},z_{2})\right])\right]\label{eq:bellman}
\end{eqnarray}
In the first expression of the right hand side, the expectation is over the policy $\pi^j$. In, the last expression, the expectation is over both the transition probability kernel for agents of type $j$ and the policy $\pi^j$ of the agents of type $j$.

\begin{rmk}\label{rmk:act_bellman}
When the transition probability and reward functions of an agent of type $j$ also depend on the action distribution of the agents of type $k,k\neq j$, the Bellman equation can also be written similarly where the value function would be a function of  additional terms $\nu_{1},\nu_{2}$ and $\pi^k$.
\end{rmk}




We define the $Q$ function for a given policy $\pi^j$ of the representative agent of type $j$ as the following
\begin{align}\label{eq:qleader}
    & Q^j(x_j,a_j,\pi^j,z_{1},z_{2})=r^j(x_j,z_{1},z_{2},a_j)+ \gamma\mathbb{E}[  V^{l}(x^{\prime},\pi^j,z_{1},z_{2})],
\end{align}
where the last expectation is over transition probability kernel and the policy for agents of type $j$. 

 
 Note from (\ref{eq:bellman}) and (\ref{eq:qleader}) that $V^lj$ is the expectation of $Q^j$ under $\pi^j$. 
\section{How to compute stationary MMFE?}\label{sec:compute}
In this section, we provide a methodology to compute the MMFE. We also characterize the conditions under which a MMFE exists. Here, the agents are completely aware of the reward functions and the transition dynamics. 


Step 1. Fix $z_{1},z_{2}$. When $z_{1},z_{2}$ are fixed, finding optimal policies become a classical optimization problem. 
Further, for a given $z_{1},z_2$, one can define a mapping from the population distribution to an optimal randomized policy $\pi^j$ among all the admissible policies $\Pi^j$ for an agent of type $j$.  Let us define this mapping be 
\begin{align}
\Gamma^j: \mathcal{P}(\mathcal{X}_1\times \mathcal{X}_2)\rightarrow \Pi
\end{align}
Here, the mapping is from the joint state spaces of both the types of the agents to the policy space. 

We assume the following
\begin{assum}\label{assum1}
There exist a constant $d_1$ such that for any $z_1^{\prime},z_2^{\prime}$ 
\begin{align}
D(\Gamma^1(z_1,z_2),\Gamma^1(z_1^{\prime},z_2^{\prime}))\leq d_1W_1((z_1,z_2),(z_1^{\prime},z_2^{\prime}))\nonumber\\
D(\Gamma^2(z_1,z_2),\Gamma^2(z_1^{\prime},z_2^{\prime}))\leq d_1W_1((z_1,z_2),(z_1^{\prime},z_2^{\prime}))
\end{align}
where 
\begin{align}
D(\pi,\pi^{\prime})=\sup_{y}W_1(\pi(y),\pi^{\prime}(y))
\end{align}
and $W_1$ is the $l_1$-Wasserstein distance between two probability measures.
\end{assum}

Step 2. Based on the analysis on Step 1, update the sequence $z_{j}$ to $z_j^{\prime}$ for $j=1,2$ according to population dynamics in (\ref{eq:mean}). 

Accordingly, for any admissible policy of an agent of type $j$ $\pi^j$ and the joint population distribution $(z_1,z_2)$, define the mapping $\Gamma_2:\Pi^1\times \Pi^2\times P(\mathcal{X}_1\times\mathcal{X}_2)\rightarrow P(\mathcal{X}_1\times \mathcal{X}_2)$. Basically, $\Gamma_2$ operates on the optimal policy of the agents of both the types for a given population distribution and outputs a population distribution according to (\ref{eq:mean}) based on the policies of the agents of both the types and the current population distribution.

We assume the following
\begin{assum}\label{assum2}
There exist constants $d_2,d_3\geq 0$, such that for any admissible policies $\pi^j, \pi^{j,\prime}$  and population distributions $z_1,z_1^{\prime},z_2,z_2^{\prime}$, 
\begin{eqnarray}
& W_1(\Gamma_2(\pi^{1,\prime},\pi^{2,\prime},(z_1,z_2)),\Gamma_2(\pi^1,\pi^2,(z_1,z_2)))\leq\nonumber\\ & d_2\max(D(\pi^{1,\prime},\pi^1),D(\pi^{2,\prime},\pi^2))\nonumber\\
& W_1(\Gamma_2(\pi^1,\pi^2,(z_1,z_2)),\Gamma_2(\pi^1,\pi^2,(z_{1}^{\prime},z_2^{\prime})))\leq\nonumber\\ & d_3W_1((z_1,z_2),(z_1^{\prime},z_2^{\prime}))\nonumber
\end{eqnarray}
\end{assum}

Step 3. Repeat Steps 1 and 2 until both $z_1=z_1^{\prime}$ and $z_2=z_2^{\prime}$. 

The following theorem shows that the prescribed methodology indeed converges to stationary MMFE.
\begin{theorem}
Given Assumptions~\ref{assum1} and \ref{assum2}, and assume $d_1d_2+d_3< 1$, then there exists a stationary MMFE. 
\end{theorem}
\begin{proof}
By the definition of stationary MMFE, $(\pi^1,\pi^2,z_1,z_2)$ is a stationary MMFE, if and only if $(z_1,z_2)=\Gamma(z_1,z_2)$ \\   $=\Gamma_2(\Gamma^1(z_1,z_2),\Gamma^2(z_1,z_2), z_1,z_2)$ and $\pi^j=\Gamma^j(z_1,z_2)$ where $\Gamma(z_1,z_2)=\Gamma(\Gamma^1(z_1,z_2),\Gamma^2(z_l,z_f))$.

This indicates for any two population distributions $z_1^{\prime},z_1,z_2^{\prime},z_2$
\begin{align}
& W_1(\Gamma(z_1,z_2),\Gamma(z_1^{\prime},z_2^{\prime}))\nonumber\\
& = W_1(\Gamma_2(\Gamma^1(z_1,z_2),\Gamma^2(z_1,z_2), z_1,z_2),\Gamma_2(\Gamma^1(z_1^{\prime},z_2^{\prime}),\Gamma^2(z_1^{\prime},z_2^{\prime}), z_1^{\prime},z_2^{\prime})\nonumber\\
& \leq W_1(\Gamma_2(\Gamma^1(z_1,z_2),\Gamma^2(z_1,z_2), z_1,z_2), \Gamma_2(\Gamma^1(z_1^{\prime},z_2^{\prime}),\Gamma^2(z_1^{\prime},z_2^{\prime}),z_1,z_2))\nonumber\\
& +W_1(\Gamma_2(\Gamma^1(z_1^{\prime},z_2^{\prime}),\Gamma^2(z_1^{\prime},z_2^{\prime}),z_1,z_2),\Gamma_2(\Gamma^1(z_1^{\prime},z_2^{\prime}),\Gamma^2(z_1^{\prime},z_2^{\prime}),z_1^{\prime},z_2^{\prime}))\nonumber\\
& \leq d_2d_1W_1((z_1,z_2),(z_1^{\prime},z_2^{\prime}))+d_3W_1((z_1,z_2),(z_{1}^{\prime},z_2^{\prime}))
\end{align}
Since $0<d_1d_2+d_3<1$, by Banach's fixed point theorem, there exists a stationary MMFE. 
\end{proof}


\begin{rmk}
The Assumption 1 can be represented in a more explicit form for certain type of reward function (such as quadratic in action, action space is convex) similar to \cite{2019fitted} which considers only one type of agents. When the action space is finite Assumption 2 can written in the following lemma.
\end{rmk}
\begin{lemma}
Suppose that $\max_{j}\tau^j(x^{j,\prime}|x_j,a_j,z_1,z_2)\leq c_1$ for all $z_1,z_2$, $x_j,a_j,x^{j,\prime}$, and $\tau^j$, $j=1,2$ are $c_2$-Lipschitz in $W_1$, i.e., $\forall x_j,x^{\prime}_j,a_j$
\begin{align}
    |\tau^j(x^{\prime}_j|x_j,a_j,z_1,z_2)-\tau^l(x^{\prime}_j|x_j,a_j,z^{\prime}_{1},z_2^{\prime})|\leq c_2W_1((z_1,z_2),(z_1^{\prime},z_2^{\prime}))\nonumber
\end{align}
then in Assumption~\ref{assum2}, $d_2$ and $d_3$ are
\begin{align}
    d_2=\dfrac{\max_j diam(\mathcal{X}_j)\max_j|\mathcal{X}_j|c_1}{\min_jd_{min}(\mathcal{A}_j)}\nonumber\\
    d_3=\dfrac{\max_j diam(\mathcal{X}_j)c_2}{2}
\end{align}
where $d_{min}(\mathcal{A}_j)=\min_{a_j\neq a_j^{\prime}}||a_j-a_j^{\prime}||$,  $diam(\mathcal{X}_j)=\max_{x_j\neq x^{\prime}_j\in \mathcal{X}_j}||x_j-x^{\prime}_j||$ which are non-zero for finite action and state space. 
\end{lemma}
Note that since we need $d_1d_2+d_3<1$, thus, the value of $c_2$ must be small. Intuitively, if the transition probabilities do not change much with the change in the population distribution, the optimal policy will also not change much which leads to a convergence. 

Further, if there is a gap between any two actions, the optimal policy would not change much even when the population distributions change. 
\section{Reinforcement Learning Based Algorithm}\label{sec:rl}
In this section, we provide a reinforcement learning based algorithm to find policies for agents when they are unaware of the reward and the transition probability dynamics. 
\subsection{Proposed Algorithm}
We propose a policy gradient based algorithm to find a MMFE. The policy gradient mechanism is a model free approach. {\em Compared to \cite{mahajan}, we characterize the conditions, which are easier to verify for a policy gradient mechanism to converge}.  We, first, characterize a simulator which is essential for the algorithm.

\textbf{Simulator}: We assume that the algorithm has access to a simulator which would give samples of next states according to the transition probability $\tau^j(\cdot|x_j,a_j,z_1,z_2)$  and the corresponding rewards for both $j=1,2$. This is a standard assumption in the literature for computing optimal RL algorithm. 

In the policy gradient  mechanism, we represent the policies of agents of type $j$, $\pi^j$ as parameterized by the parameter $\theta_j$.  Specifically, $\pi_{\theta_j}$ gives a choice of action for an agent of type $j$ for every state.  One of the classical examples is Boltzman policy. In the Boltzman policy, $\pi_{\theta_j}$ gives the probability of taking action $a_j$ when the state is $a_j$ via the following function $\dfrac{e^{h(x_j,a_j,\theta_j)}}{\sum_{a_k}e^{h(x_j,a_k,\theta_j)}}$  where $h(x_j,a_k,\theta)=\sum_{i}\theta_{j,i}f^j_i(x_j,a_k)$ and $f^j_i$ is a function of the state $x_j$ and action $a_k$ for agents of type $j$. 

Now, we specify an algorithm (Algorithm~\ref{alg1}) to provide an unbiased estimator of estimating the $Q$-function for the agents  which we then use it to define our policy gradient mechanism. 
\begin{algorithm}[!t]
%
	\caption{EST-Q:Unbiasedly estimating $Q$-function for agent type $j=1,2$}
	\label{alg1}
	\begin{algorithmic}
		\State \textbf{Input: } 
		$s$: Initial State, 
		$a$: Action, 
		$\theta_j$: Policy Parameter,
		$z_{1},z_{2}$: Population distribution,
		The population simulator
		\State \textbf{Output: } $Q$- function for a given policy
		
		\State Initialize $\hat{Q}\rightarrow 0, s_0\rightarrow 0, a_0\rightarrow a$.
		\State Draw $T$ from a Geometric Distribution with parameter $1-\gamma^{1/2}$, i.e., $\Pr(T=t)=(1-\gamma^{1/2})\gamma^{t/2}$.
		\For{$t=0, \cdots T-1$}
		\State Collect and add the instantaneous reward $r^j(s_t,a_t,z_{1},z_{2})$ to $\hat{Q}$, $\hat{Q}\rightarrow \hat{Q}+\gamma^{t/2} r^j(s_t,a_t,z_{1},z_{2})$.
		\State Simulate the next state according to a population simulator and action $a_{t+1}\sim \tilde{\pi}_{\theta_j}(\cdot| s_t,a_t)$		
	        \EndFor
		\State Collect $r^j(s_T,a_T,z_{1},z_{2})$ and update  $\hat{Q}=\hat{Q}+\gamma^{T/2}r^j(s_T,a_T,z_1,z_2)$.
		
		\Return $\hat{Q}$.
	\end{algorithmic}
\end{algorithm}
Algorithm~\ref{alg1} provides an unbiased estimator of $Q^j(x,a,\pi_{\theta_j},z_{l},z_{f})$. It gives an unbiased estimator because of the random-horizon setting via Monte-carlo rollout \cite{koppel}. Note that since the horizon $T$ is chosen randomly from a geometric distribution of parameter $1-\gamma^{1/2}$, the probability that the horizon is of length at least $t$ is given by $\gamma^{t/2}$. Now, combined with the fact that the reward is multiplied by $\gamma^{t/2}$ in the algorithm, we obtain the unbiased estimator of the $Q$-function.  While in practice, usually finite {\em deterministic} horizon roll-outs are used to estimate  infinite $Q$-function. However, it would create bias in estimating $Q$-function and hence it would end-up in creating biases in policy gradient method. In \cite{mahajan}, authors also relied on unbiased estimator of $Q$-function without explicitly stating how to find it.


Note that in the MMFE, the policy must be optimal in the value function starting from each initial state. However, it is computationally difficult to verify whether a policy is optimal for  starting from each initial state. Rather, we characterize the MMFE which provides an optimal policy for agents for the expected value function where the expectation is taken with respect to the population dynamics. Thus, we seek to compute policy such that
\begin{itemize}
\item $\pi_{\theta_j}$ maximizes, $J^j(\theta_j,z_1,z_2)=\mathbb{E}_{x_j\sim z_j}[V^j(x_j,\pi_{\theta_j},z_1,z_2)]$ for $j=1,2$
\item $(z_1,z_2)$ both satisfy (\ref{eq:state_consist}).
\end{itemize}
Since $\pi_{\theta_j}$ maximizes $J^j$, thus, from the first order of stationary condition, $\nabla J^f(\theta^{*}, z_l,z_f)=0$ and $\nabla J^j(\theta_j^{*},z_1,z_2)=0$.

From the policy-gradient method, we can represent the gradient of $J^j$   in the following form for $j=1,2$
\begin{align}\label{eq:follow_grad}
&\nabla J^j(\theta_j,z_1,z_2) =\sum_{y_0}z_j(y_0)\sum_{t=0}^{\infty}\gamma^t\sum_{y}\rho^j(y_{t}=y|y_0,\pi_{\theta_j},z_1,z_2)\nonumber\\
&\sum_{a}\nabla \pi_{\theta_j}(a|y) Q^j(y,a,\pi_{\theta_j},z_1,z_2)\nonumber\\
& =\sum_{y_0}z_j(y_0)\sum_{t=0}^{\infty}\gamma^t\sum_{y}\rho^j(y_{t}=y|y_0,\pi_{\theta_j},z_1,z_2)\nonumber\\
&\sum_{a}\pi_{\theta_j}(a|y)\nabla \log\pi_{\theta_j}(a|y)Q^j(y,a,\pi_{\theta_j},z_1,z_2)
\end{align}
where 
\begin{equation}
\rho^j(y_t=y|y_0,\pi_{\theta_j},z_1,z_2)=\sum_{k=0}^{t-1}\sum_{b^{k}}\pi_{\theta_j}(b^{k}|y_{k})\tau^j(y_{k+1}|y_{k},b^{k},z_1,z_2)\nonumber
\end{equation}
i.e., the probability that the state is $y$ at time $t$ for agent of type $j$ under the policy $\pi_{\theta_j}$ and when it starts from the state $y_0$. The first expression comes from the standard policy gradient mechanism. The second equality comes from the fact that $\nabla \log(x)=\dfrac{\nabla{x}}{x}$.


Now, we describe the policy gradient mechanism in order to find policies for leaders and followers,
\begin{algorithm}[!t]
%
	\caption{RHPG-MMFE: Random Horizon Policy gradient to find MMFE}
	\label{alg1}
	\begin{algorithmic}
		\State \textbf{Input: } Population simulator
		\State \textbf{Output: } $\pi_{\theta_j}$, $j=1,2$. 
			\State \textbf{Initialization: }
			$z_{1,0},z_{2,0}$: Initial Population distribution uniform distribution\\
		
			$m\leftarrow 0$, $z_{1,m}\leftarrow z_{1,0}, z_{2,m}\leftarrow z_{2,0}$
		\While{Until Convergence}
		\State \textbf{Initialization }$\theta_{j,k}$ set at $0$, $k\leftarrow 0$, for $j=1,2$
		\While {Until Convergence}
	        \State Sample $x_{j,k}$ from $z_{j,m}$ for $j=1,2$
	        \State $x_{j,0}\leftarrow x_{j,k}$ for $j=1,2$
		\State Draw $T_{k+1}$  from Geometric Distribution with parameter $(1-\gamma)$, $a_{j,0}\sim \pi_{\theta_{j,k}}(\cdot|x_{j,0})$ for $j=1,2$.
		\For {$t=0,1,\ldots, T_{k+1}-1$}
		\State Simulate the next state $x_{j,t+1}\sim\ \tau^j(\cdot|x_{j,t},a_{j,t},z_{1,m},z_{2,m})$ for $j=1,2$  from the simulator. 
		\State Simulate the action $a_{j,t+1}\sim\pi_{\theta_{j,k}}(\cdot|x_{j,t+1})$ for $j=1,2$.
		\EndFor
		\State Obtain an Estimate of $Q^j_(x_{j,T_{k+1}},a_{j,T_{k+1}},\pi_{\theta_{j,k}},z_{1,m},z_{2,m})$ for $j=1,2$ using Algorithm 1, i.e., 
		\State $\hat{Q}^1(x_{1,T_{k+1}},a_{1,T_{k+1}},\pi_{\theta_{1,k}},z_{1,m},z_{2,m})$=EST-Q$(x_{1,T_{k+1}},a_{1,T_{k+1}},\pi_{\theta_{j,k}},z_{1,m},z_{2,m})$.
		\State $\hat{Q}^2(x_{2,T_{k+1}},a_{2,T_{k+1}},\pi_{\theta_{2,k}},z_{1,m},z_{2,m})$=EST-Q$(x_{2,T_{k+1}},a_{2,T_{k+1}},\pi_{\theta_{2,k}},z_{1,m},z_{2,m})$
		\State $\theta_{j,k+1}\leftarrow \theta_{j,k}+\dfrac{\alpha_k\hat{Q}^j(\cdot)\nabla\log(\pi_{\theta_{j,k}}(a_{j,T_{k+1}}|x_{j,T_{k+1}}))}{1-\gamma}$ for $j=1,2$.
		\State Update Iteration $k\leftarrow k+1$
		\EndWhile Until convergence
%
\For {$j=1,2$}
		\For {$l=1,\ldots,N_j$}
		
				\State Sample $x_{j,l}$ from $z_{j,m}$  take action $a_j\sim\pi_{\theta_{j,k}}(\cdot|x_{j,l})$,  simulate the next state $x_{j,m+1}$ . 
				\State $z_{j,m+1}(x)=\dfrac{1}{N_j}\mathbbm{1}(x_{j,l}=x)$ for $j=1,2$
				\EndFor
				\EndFor
				\State Update outer iteration counter $m\leftarrow m+1$.
		\EndWhile
	\end{algorithmic}
\end{algorithm}
As an extension of the result of \cite{koppel}, we obtain
\begin{lemma}
For a given $z_1,z_2$,let
\begin{align}\label{eq:grad_nab}
\nabla \hat{J}^j(\theta_j)=\dfrac{1}{1-\gamma}\hat{Q}^j(x_{j,T},a_{j,T},\pi_{\theta_j}z_1,z_2)\nabla\log(\pi_{\theta_j}(a_T|x_T))
\end{align}
for $j=1,2$. Then, for any $\theta_j$, for$j=1,2$,
\begin{align}
\mathbb{E}[\nabla \hat{J}^j(\theta_j,z_l,z_f)|\theta_j]=\nabla J^j(\theta_j)
\end{align}
where the expectation is taken over the random sample from $z_1$ and $z_2$, random horizon $T^{\prime}$, the trajectory along $(x_{1,0},a_{1,0},x_{2,0},a_{2,0}\ldots,\\x_{1,T^{\prime}_1},a_{1,T^{\prime}_1},x_{2,T^{\prime}},a_{2,T^{\prime}})$.
\end{lemma}
The above lemma indicates that $\nabla \hat{J}^j$ for $j=1,2$ is an unbiased estimator of $\nabla{J}^j$.  Intuitively, from Algorithm 1 note that $\hat{Q}^j_(x_{j,T},a_{j,T},\\\pi_{\theta_j},z_1,z_2)$  gives an unbiased estimator of $Q^j(x_{j,T},a_{j,T},\pi_{\theta_j},z_1,z_2)$. Thus, from (\ref{eq:follow_grad})  it follows that $\nabla\hat{J}^j$   gives unbiased estimators of $\nabla{J}^j$ for $j=1,2$.

In Algorithm 2, we first fix $z_1,z_2$ and try to find the optimal value functions for leaders and followers for the given $z_1,z_2$ similar to Step 1 in Section~\ref{sec:compute}. The parameters are updated using these unbiased estimators in Algorithm 2 and thus, it would converge when a stationary point exists. Once the inner loop converges, we compute new $z^{\prime}_{1}$ and $z_2^{\prime}$ based on the optimal policies acted upon $z_1$ and $z_2$ similar to Step 2. When $z_{1}^{\prime}$ and $z_{2}^{\prime}$ become close to $z_l, z_f$ respectively, we stop the algorithm similar to Step 3.

In order to find optimal policy for a given $z_l,z_f$, note that first a random horizon is selected with Geometric distribution and given the policy and population simulator, next state and action are obtained both for leader and follower. This is done to get an unbiased estimate of $\gamma^t\rho^j(y_t=y|y_0,\pi_{\theta_{j,k}},z_1,z_2)$. After that the algorithm obtains the unbiased estimator of $Q$-function from Algorithm 1. Finally, the parameters are updated using the unbiased estimator of the gradient. 
\subsection{Convergence of Policy Gradient}
In  this section, we prove the convergence of the policy gradient mechanism to a MMFE under certain additional assumptions. 

\begin{assum}\label{assum3}
\begin{itemize}
\item The reward function is bounded 
\item $\pi_{\theta_j}$ is differentiable with respect to $\theta_j$ for $j=1,2$. Further, for any $(x_j,a_j)$ for $j=1,2$
\begin{align}
& ||\nabla\log(\pi_{\theta^1_{j}}(a_j|x_j))-\nabla\log(\pi_{\theta^2_{j}}(a_j|x_j))||\leq L_{\theta}||\theta^1_j-\theta^2_j||\quad\text{for any }\theta^1_j,\theta^2_j \nonumber\\
& ||\nabla\log(\pi_{theta_j}(a_j|x_j))||\leq B_{\theta}\nonumber
\end{align}
\item $\alpha_k$ is such that 
\begin{align}
\sum_{k=0}^{\infty}\alpha_k=\infty,\quad \sum_{k=0}^{\infty}\alpha_k^2<\infty.\nonumber
\end{align}
\end{itemize}
\end{assum}
A lot of policies satisfy the second condition including Boltzman-policy and Gaussian policies. In the simulation, we choose $\alpha_k=k^{-a}$ for some constant $a\in (1/2,1)$. 

First, we show the convergence of the random horizon policy gradient mechanism to the optimal policy,,
\begin{theorem}\label{thm:pol}
Under Assumption~\ref{assum3}, as $k\rightarrow \infty$, $\theta_{j,k}\rightarrow \theta^*_j$,  where $\theta_j^{*}$ is are the optimal policy parameter for agents of type $j$ for $j=1,2$ for a given $z_{1,m},z_{2,m}$. 
\end{theorem}
The proof readily follows from Theorem 3.4 in \cite{koppel}, thus, we omit it here. The above theorem indeed shows that the algorithm finds the optimal policies for a given population distribution of both the types (similar to step 1). 

Equipped with the result, we have the following
\begin{theorem}
Under Assumptions~\ref{assum1},\ref{assum2} and \ref{assum3}, Algorithm 2 converges to the MMFE policy profile as $m\rightarrow \infty$, $N_1, N_2\rightarrow \infty$ and $k\rightarrow \infty$. 
\end{theorem}

\begin{rmk}
Though the $Q$-estimator given by Algorithm 1 is unbiased, it may have high variance. We can reduce the variance using baseline method by using value function estimator. The analysis would be similar, and thus, we omit it here. 

Note that Monte-Carlo roll out may not be sample efficient, instead, one can use Bootstrapping method such as Actor-Critic method which is sample amiable. However, such a method may lead to an {\em biased} estimator of the $Q$-function. Developing a policy gradient algorithm using Actor-Critic method for computing MMFE is left for the future. 

Further, one can also employ function approximation techniques when the state-space and action spaces become large. 
\end{rmk}
Note that in practice, the  next state population distributions are obtained from the number of agents of each type (i.e., $N_j$ is equal to the number of agents of type $j$). When either  $N_1$ or $N_2$ is finite, the equilibrium error is at most $\min_{j}O(1/N_j)$   \cite{huang2006large,weintraub2006oblivious}. 

\section{Numerical Simulations}\label{sec:numerical}
\subsection{Cyber Attack}
\subsubsection{Model}
We simulate a cyber-attack model with large number of defenders  and attackers as a MFG with two different types of agents. 

A defender's state can be discretized into $n+1$ states $0, \frac{1}{n}, \cdots, 1$. A higher state means more vulnerable. Each defender can take action $a_t\in \{0,1\}$. If $a_t=0$, the next state is $0$, and it would be secure. If $a_t=1$, the defender does not take any action.

The state evolution model of the defender $i$ is--
\begin{align}
x_{t+1}^i=\begin{cases} x_t^i+w_{1,i}\quad \text{if }a_t^i=1\nonumber\\
0\text{ otherwise }\end{cases}
\end{align}
 where $w_{1,i}$ is a random variable with identical masses on $\{0,1/n,\ldots,1-x_t^i\}$.  
 
 The attacker's state is also discretized from $0, \frac{1}{n}, \cdots, 1$, in $n+1$ states. A lower state means more powerful. An attacker can successfully attack a defender if the attacker's state is lower than the defender. The defender can take an action $b_t=0$ if it wants to become extremely powerful, i.e., the next state would be $0$. 
 
 The  defender $i$'s state evolves as the following
 \begin{equation}
 y_{t+1}^i=\begin{cases} y_t^i+w_{2,i},\quad \text{if } b_t^i=1\nonumber\\
 0\text{ otherwise } \end{cases}
 \end{equation}
 where $w_{2,i}$ is a random variable with uniform mass on $\{0,1/n,\cdots, 1-y_t^i\}$.

Now, the reward model can be described as the following-- for a defender, 
\begin{align}
& r^1(x_t,a_t,z_{1,t},z_{2,t})=-g_1x_t-\nonumber\\& g_2x_t(\max\{\mathbb{E}(z_{1,t})-\mathbb{E}(z_{2,t}),0\})-(1-a_t)\lambda_1\nonumber
\end{align}
 $\lambda_1$ is the cost of securing the node. Note that a leader will be more vulnerable when its state is high.  Further, if the mean value of the states of the followers is small it means that there are more powerful attackers which can attack the defenders easily. Thus, a defender can also be infected if the mean state of the attackers (i.e.,$\mathbb{E}(z_{2,t}$) is smaller compared to the mean state of the defenders ($\mathbb{E}(z_{1,t})$).  

Likewise, the reward model for an attacker is
\begin{align}
& r^2(y_t,b_t,z_{1,t},z_{2,t})=-g_1y_t-\nonumber\\&g_2y_t(\max\{\mathbb{E}(z_{2,t})-\mathbb{E}(z_{1,t}),0\})-(1-b_t)\lambda_2
\end{align}
$\lambda_2$ is the cost for attacker to become most powerful. When the state of the follower is high, it would attack a defender, however, it would have smaller chance to breach the defender. On the other hand, it the mean state of the defenders is smaller compared to the attackers, the attackers may need to attack more number of defender since the virus propagation from one defender to another defender is small which deplete the energies of the attackers. 

\subsubsection{Set-up}
In our evaluation, we consider that there are $N_1=100$ number of defenders and $N_2=100$ number of attackers. We set $n=10$. We set $g_1=0.2, g_2=0.1, \lambda_1=\lambda_2=0.5$. The parameterized policy we consider is the Boltzman policy where action $0$ is taken with probability $1/(1+\exp(\theta_{1,1}x_1+\theta_{1,2}(1-x_1)))$ for defender at state $x$. Note that as $\theta_{1,1}$ decreases defenders take action $0$ with a higher probability as the state value $x$ increases. 

Similar to the defender, the attacker's policy is also considered to be of the following form-- at state $y$, action $0$ is taken with probability $1/(1+\exp(\theta_{2,1}y+\theta_{2,2}(1-y)))$.

\subsubsection{Results}
Fig.~\ref{fig:policy} shows the probabilities of taking action $a_t=0$ ($b_t=0$, resp.) by defender (attacker, resp.) as a function of state. As the state increases, the probability increases. This is because higher state indicates a lower reward for both defender and attacker. A defender will be more vulnerable to the attacker if the state is higher. On the other hand, higher state of an attacker indicates that the attacker can not successfully attack a defender. Thus, one takes action in order to achieve higher reward in the future. 

Fig.~\ref{fig:policy} shows a threshold type of behavior for both attacker and defender. When the state is less than $0.3$, the probability of taking any action is negligible for both attacker and defender.On the other hand, when the state exceeds $0.5$ both the attacker and defender take actions with certainty. When the state is $0.3$, the attacker takes an action with slightly higher probability because of slightly higher mean value of the state of attackers ($3.997$) compared to the defender ($3.955$). This slight variation is due to the parameter used to check convergence. Note that because of the symmetric nature of the game, both policies converge to almost same value. 

\begin{figure}
    \centering
    \includegraphics[width=80mm]{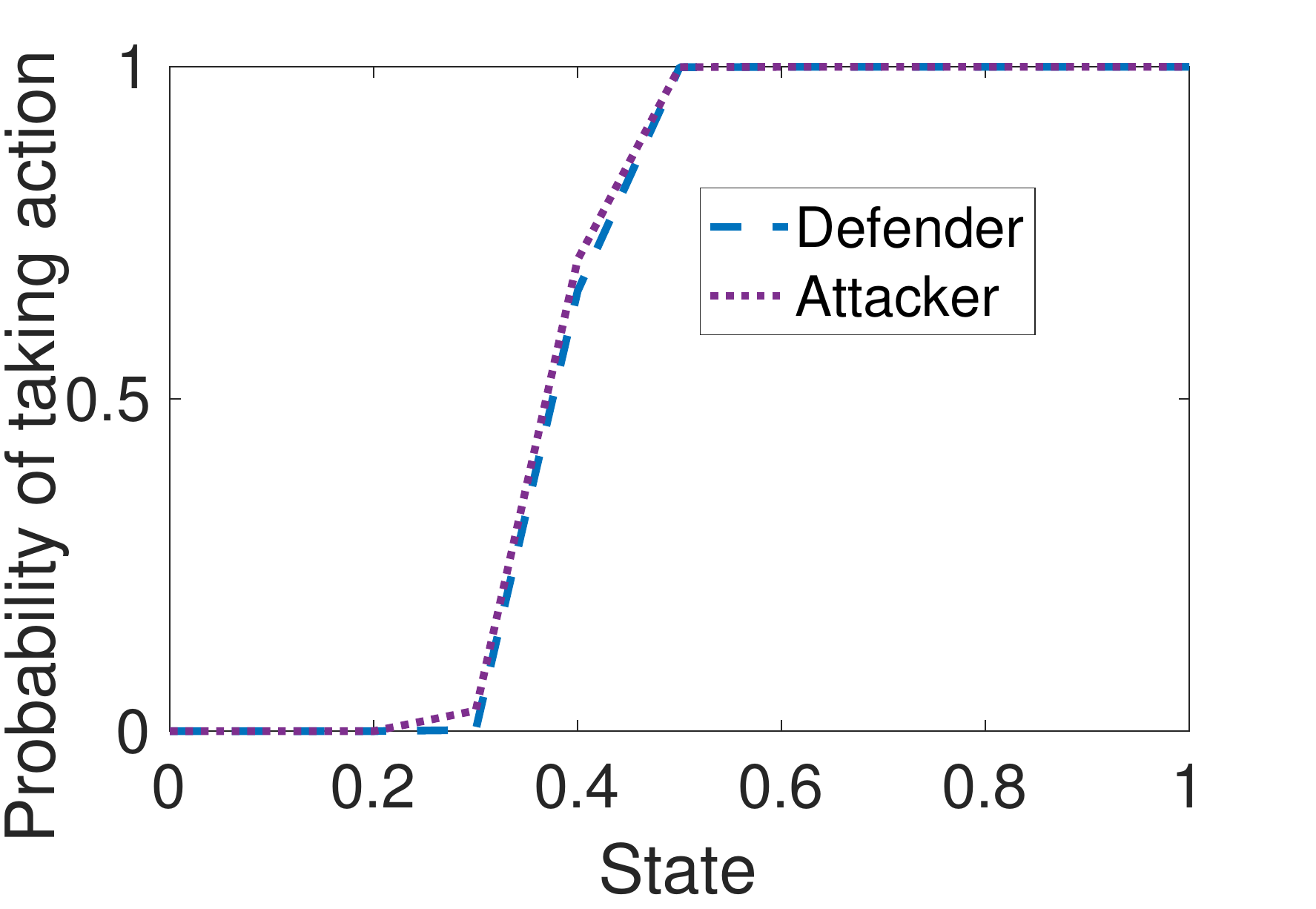}
    \vspace{-0.2in}
    \caption{Policies of defenders and attackers
    as function of states}
    \label{fig:policy}
    \vspace{-0.2in}
\end{figure}

\section{Discussion}\label{sec:discussion}
\subsection{MMFE with both state and action distribution}
As we have discussed in Remarks in ~\ref{rmk:act_bellman} and \ref{rm:act-mfe}, we can extend our analysis when the reward, and transition probability both depend on the action distribution of the players as well. We only need to replace $z_{j,t}$ with joint state and action distribution of the agents of type $j$ $(z_{j,t},\nu_{j,t})$ for $j=1,2$. The algorithm also converges under additional condition with replacing $z_{j}$ with $(z_{j},\nu_{j})$. 

\subsection{More than two types of agents}
We can extend the setting with more than two types of agents where each type has infinite number of players. For example, when there are more than two types of agents,the reward function and transition probabilities depend on state distribution functions of each types of agents and the individual state and action. Hence, in the stationary MMFE, one can compute the optimal policy for each type of agents while considering the other agents' distributions as fixed. After the optimal policy is found for a fixed set of population distribution for each set of agents, we update the population state distributions and we repeat until we reach the convergence. 

\subsection{Non-stationary MMFE}
Similar to the stationary MMFE, a non-stationary MMFE can also be obtained. First,  we fix the sequence of population state distributions $z_{j,t}$ for $j=1,2$ and $t=0,\cdots,\infty$ and obtain the optimal policies for the agents of both the types (Step 1). We then update the population state distribution $z_{j,t}^{\prime}$ across all the time-horizon for $j=1,2$. If $z_{1,t}^{\prime}$ and $z_{2,t}^{\prime}$ do not match with $z_{1,t}$ and $z_{2,t}$ respectively for every $t$, we then again obtain policies for the updated population state distributions $z_{j,t}^{\prime},$ (Step 2). We repeat the process until $z_{j,t}^{\prime}$ becomes equal to $z_{j,t}$ for $j=1,2$. Computing a non-stationary MMFE is computationally challenging. The characterization of a  computationally efficient RL algorithm to compute a non-stationary MMFE is left for the future. 
\section{Conclusions}
We study a  multi-agent multi-type Markov strategic interaction over a finite time steps for large number of agents of each type. A MMFE is defined as an equilibrium concept. An algorithm for known system dynamics is proposed that achieves a stationary MMFE. The condition under which a stationary MMFE exists is also characterized. A policy gradient based RL algorithm is proposed to obtain the stationary MMFE when the players are unaware of the dynamics. 

	\appendix
	\section{Proof of Theorem 3}
	
	First, we prove the theorem by assuming that the population simulator gives exactly $z_{l,m+1}$ and $z_{f,m+1}$ from $z_{l.m}$ and $z_{f,m}$, respectively, using $\pi$ and $w$.

Consider $m=M$-th iteration, 

From Theorem~\ref{thm:pol}, there exists a $K$ such that $||\theta_{j,K}-\theta^{*}_j||\leq \delta$ a such that $||\pi_{\theta_{j,K}}-\pi_{\theta^*_j}||\leq \epsilon$. Let $(z_1^{*}, z_2^{*})$ are the MFE population distributions. Recall that $\Gamma^1$  and $\Gamma^2$ are the optimal for leaders and followers respectively. Let us denote the policy of follower and leader at $m=M$-th iteration after $K$ number of inner iterations as $\pi_{\theta_{j,K}}(z_{1,m},z_{2,m})$. Now,
\begin{align}
& W_1((z_{1,m+1},z_{2,m+1}),(z_{1,*},z_{2,*}))=\nonumber\\
& W_1(\Gamma_2(\pi_{\theta_ {1,K}}(z_{1,m},z_{2,m}),\pi_{\theta_{2,K}}(z_{1,m},z_{2,m})),\nonumber\\& \Gamma_2(\Gamma^1(z_{1,*},z_{2,*}),\Gamma^2(z_{1,*},z_{2,*})))\nonumber\\
& \leq W_1(\Gamma_2(\Gamma^1(z_{1,m},z_{2,m}),\Gamma^2(z_{1,m},z_{2,m})),\nonumber\\& \Gamma_2(\Gamma_2(\Gamma^1(z_{1,*},z_{2,*}),\Gamma^2(z_{1,*},z_{2,*}))))+\nonumber\\
& W_1(\Gamma_2(\Gamma^1(z_{1,m},z_{2,m}),\Gamma^2(z_{1,m},z_{2,m})),\nonumber\\& \Gamma_2(\pi_{\theta _{1,K}}(z_{l,m},z_{f,m}),\pi_{\theta_{2,K}}(z_{1,m},z_{2,m})))\nonumber\\
& \leq dW_1((z_{1,m},z_{2,m}),(z_{1,*},z_{2,*}))+\nonumber\\ & d_2D((\Gamma^1(z_{1,m},z_{2,m}),\Gamma^2(z_{1,m},z_{2,m})),(\pi_{\theta_ {1,K}}(z_{1,m},z_{2,m}),\pi_{\theta_{2,K}}(z_{1,m},z_{2,m})))\nonumber\\
& \leq dW_1((z_{1,m},z_{2,m}),(z_{1,*},z_{2,*}))+d_2\epsilon
\end{align}
where $d=d_1d_2+d_3$. Hence, we have
\begin{align}
&W_1((z_{1,m+1},z_{2,m+1}),(z_{1,*},z_{2,*}))\leq \nonumber\\& d^mW_1((z_{1,0},z_{2,0}),(z_{1,*},z_{2,*}))+\dfrac{d_2\epsilon(1-d^m)}{1-d}
\end{align}
Since $d<1$, thus, there exists $M$ such that $d^m\leq \epsilon$ for $m>M$, and $\dfrac{d_2(1-d^m)}{1-d}<1$, hence we have
\begin{align}
W_1((z_{1,m+1},z_{2,m+1}),(z_{1,*},z_{2,*}))\leq O(\epsilon)
\end{align}

Now, we show the result for the weak population simulator which only gives the next  state of the agents a for a given state of the agent, action, and population distribution.
Let us denote the actual population dynamics be $z^{\prime}_{1,m+1}$ and $z^{\prime}_{2,m+1}$ and the population dynamics returned by Algorithm 2 are $z_{1,m+1}$ and $z_{2,m+1}$, respectively. 

Now, 
\begin{align}\label{eq:bound}
& W_1((z_{1,m+1},z_{2,m+1}),(z_{1,*},z_{2,*}))\leq \nonumber\\&  W_1(\Gamma_2(\pi_{\theta_{1,K}}(z_{1,m},z_{2,m}),\pi_{\theta_{2,K}}(z_{1,m},z_{2,m})),\nonumber\\& \Gamma_2(\Gamma^1(z_{1,*},z_{2,*}),\Gamma^2(z_{1,*},z_{2,*})))\nonumber\\& +W_1((z_{1,m+1},z_{2,m+1}),\Gamma_2(\pi_{\theta_{1,K}}(z_{1,m},z_{2,m}),\pi_{\theta_{2,K}}(z_{1,m},z_{2,m})))\nonumber\\&
=W_1(\Gamma_2(\pi_{\theta_{1,K}}(z_{1,m},z_{2,m}),\pi_{\theta_{2,K}}(z_{1,m},z_{2,m})),\nonumber\\& \Gamma_2(\Gamma^1(z_{1,*},z_{2,*}),\Gamma^2(z_{1,*},z_{2,*})))+\nonumber\\&
W_1((z_{1,m+1},z_{2,m+1}),\mathbb{E}(z_{1,m+1},z_{2,m+1}))
\end{align}
where
\begin{align}
&\mathbb{E}(z_{1,m+1}(x^{\prime}))=\nonumber\\& \sum_{x}\sum_{a}z_{1,m}(x)\tau^1(x^{\prime}|x,\pi_{\theta_{1,K}}(a|x),z_{1,m},z_{2,m})\nonumber\\
& \mathbb{E}(z_{2,m+1}(y^{\prime}))=\nonumber\\ & \sum_y\sum_{a}z_{2,m}(y)\tau^2(y^{\prime}|y,\pi_{\theta_{2,K}}(b|y),z_{1,m},z_{2,m})
\end{align}
 From Hoeffding's inequality, we have
 \begin{align}
 \Pr(|z_{1,m+1}(x^{\prime})-\mathbb{E}(z_{1,m+1}(x^{\prime}))|>t)\leq 2\exp(-N_1t^2)\nonumber\\
 \Pr(|z_{2,m+1}(y^{\prime})-\mathbb{E}(z_{2,m+1}(y^{\prime}))|>t)\leq 2\exp(-N_2t^2)\nonumber
 \end{align}
Hence, we can bound the second expression in the right hand side of (\ref{eq:bound}) by $\epsilon$ with high probability for large enough $N_1$ and $N_2$. Since we already have obtained bound for the first expression in the right hand side. Hence, with a high probability, we prove the convergence as $K\rightarrow \infty$,  and $N_1, N_2\rightarrow \infty$. \qed

	\bibliographystyle{ACM-Reference-Format} 
	\bibliography{mybib}

\end{document}